%% file: arxivMain.tex
\documentclass[sigconf,10pt]{acmart}

\renewcommand\footnotetextcopyrightpermission[1]{}
\pdfoutput=1

\usepackage{verbatim}
\usepackage{booktabs}
\usepackage{balance}
\usepackage{amsmath}
\usepackage{algorithmicx}
\usepackage[plain]{algorithm}
\usepackage{algpseudocode}
\usepackage{pgfplots}
\usepackage{pgfplotstable}
\usepgfplotslibrary{groupplots}
\usepackage[center]{subfigure}
\usepackage{url}
\usetikzlibrary{patterns,shapes,arrows,positioning,fit,decorations.pathreplacing}
\usetikzlibrary{pgfplots.external}
\pgfplotsset{compat=1.3}

\author{Immanuel Trummer, Junxiong Wang, Deepak Maram, Samuel Moseley, Saehan Jo, Joseph Antonakakis}
\email{{itrummer,jw2544,sm2686,sjm352,sj683,jma353}@cornell.edu}
\affiliation{%
  \institution{Cornell University, Ithaca (NY)}
}

\begin{document}

\fancyhead{}

\title{SkinnerDB: Regret-Bounded Query Evaluation\\via Reinforcement Learning}

\begin{abstract}
SkinnerDB is designed from the ground up for reliable join ordering. It maintains no data statistics and uses no cost or cardinality models. Instead, it uses reinforcement learning to learn optimal join orders on the fly, during the execution of the current query. To that purpose, we divide the execution of a query into many small time slices. Different join orders are tried in different time slices. We merge result tuples generated according to different join orders until a complete result is obtained. By measuring execution progress per time slice, we identify promising join orders as execution proceeds. 

Along with SkinnerDB, we introduce a new quality criterion for query execution strategies. We compare expected execution cost against execution cost for an optimal join order. SkinnerDB features multiple execution strategies that are optimized for that criterion. Some of them can be executed on top of existing database systems. For maximal performance, we introduce a customized execution engine, facilitating fast join order switching via specialized multi-way join algorithms and tuple representations. 

We experimentally compare SkinnerDB's performance against various baselines, including MonetDB, Postgres, and adaptive processing methods. We consider various benchmarks, including the join order benchmark and TPC-H variants with user-defined functions. Overall, the overheads of reliable join ordering are negligible compared to the performance impact of the occasional, catastrophic join order choice.
\end{abstract}

\maketitle

\section{Introduction}
\label{introSec}
\input{sectionsPdf/intro.tex}

\input{sectionsPdf/related.tex}
\input{sectionsPdf/overview.tex}
\input{sectionsPdf/algorithm.tex}
\input{sectionsPdf/analysis.tex}


\section{Implementation and Evaluation}
\label{experimentsSec}
\input{sectionsPdf/experiments.tex}

\section{Conclusion}
\label{conclusionSec}
\input{sectionsPdf/conclusion.tex}


\bibliographystyle{ACM-Reference-Format}
\bibliography{../../library}

\appendix

\section{Additional Experiments}
\label{experimentsApp}
\input{sectionsPdf/experimentsApp.tex}

\end{document}

%% file: sectionsPdf/intro.tex

\begin{center}
\textit{``The consequences of an act affect the probability of its occurring again.''} --- B.F. Skinner.
\end{center}

Estimating execution cost of plan candidates is perhaps the primary challenge in query optimization~\cite{Lohman2014}. Query optimizers predict cost based on coarse-grained data statistics and under simplifying assumptions (e.g., independent predicates). If estimates are wrong, query optimizers may pick plans whose execution cost is sub-optimal by orders of magnitude. We present SkinnerDB, a novel database system designed from the ground up for reliable query optimization.


SkinnerDB maintains no data statistics and uses no simplifying cost and cardinality models. Instead, SkinnerDB learns (near-)optimal left-deep query plans \textit{from scratch} and on the fly, i.e.\ \textit{during} the execution of a given query. This distinguishes SkinnerDB from several other recent projects that apply learning in the context of query optimization~\cite{Krishnan2018, Marcus2018}: instead of learning from past query executions to optimize the next query, we learn from the current query execution to optimize the remaining execution of the current query. Hence, SkinnerDB does not suffer from any kind of generalization error across queries (even seemingly small changes to a query can change the optimal join order significantly). 

SkinnerDB partitions the execution of a query into many, very small time slices (e.g., tens of thousands of slices per second). Execution proceeds according to different join orders in different time slices. Result tuples produced in different time slices are merged until a complete result is obtained. After each time slice, execution progress is measured which informs us on the quality of the current join order. At the beginning of each time slice, we choose the join order that currently seems most interesting. In that choice, we balance the need for exploitation (i.e., trying join orders that worked well in the past) and exploration (i.e., trying join orders about which little is known). We use the UCT algorithm~\cite{Kocsis2006} in order to optimally balance between those two goals.

\begin{figure}[t]
\includegraphics{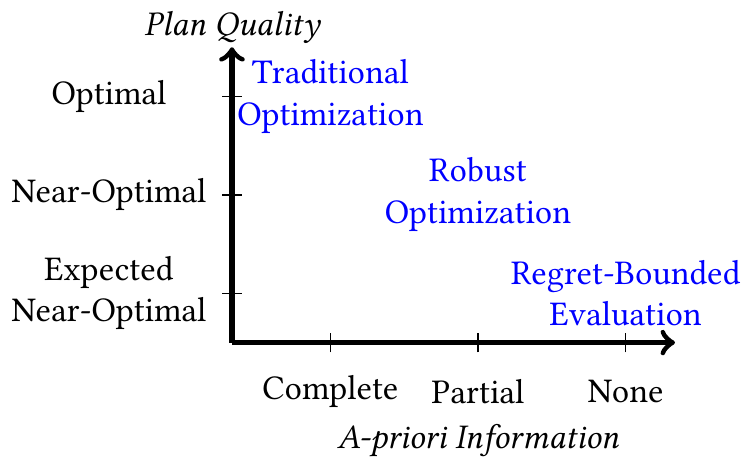}
\caption{Tradeoffs between a-priori information and guarantees on plan quality in query evaluation.\label{guaranteesVsAssumptionsFig}}
\end{figure}

Along with SkinnerDB, we introduce a new quality criterion for query evaluation methods. We measure the distance (additive difference or ratio) between expected execution time and execution time for an optimal join order. This criterion is motivated by formal regret bounds provided by many reinforcement learning methods. In the face of uncertainty, based on minimal assumptions, they still bound the difference between expected and optimal average decision quality. Traditional query optimization guarantee optimal plans, provided that complete information (e.g., on predicate selectivity and predicate correlations) is a-priori available. We assume that no a-priori information is available at the beginning of query execution (see Figure~\ref{guaranteesVsAssumptionsFig}, comparing different models in terms of assumptions and guarantees). Our scenario matches therefore the one considered in reinforcement learning. This motivates us to apply a similar quality criterion. The adaptive processing strategies used in SkinnerDB are optimized for that criterion.

SkinnerDB comes in multiple variants. Skinner-G sits on top of a generic SQL processing engine. Using optimizer hints (or equivalent mechanisms), we force the underlying engine to execute specific join orders on data batches. We use timeouts to limit the impact of bad join orders (which can be significant, as intermediate results can be large even for small base table batches). Of course, the optimal timeout per batch is initially unknown. Hence, we iterate over different timeouts, carefully balancing execution time dedicated to different timeouts while learning optimal join orders. Skinner-H is similar to Skinner-G in that it uses an existing database management system as execution engine. However, instead of learning new plans from scratch, it partitions execution time between learned plans and plans proposed by the original optimizer. 

Both, Skinner-G and Skinner-H, rely on a generic execution engine. However, existing systems are not optimized for switching between different join orders during execution with a very high frequency. Skinner-C exploits a customized execution engine that is tailored to the requirements of regret-bounded query evaluation. It features a multi-way join strategy that keeps intermediate results minimal, thereby allowing quick suspend and resume for a given join order. Further, it allows to share execution progress between different join orders and to measure progress per time slice at a very fine granularity (which is important to quickly obtain quality estimates for join orders). 

In our formal analysis, we compare expected execution time against execution time of an optimal join order for all Skinner variants. For sufficiently large amounts of input data to process and under moderately simplifying assumptions, we are able to derive upper bounds on the difference between the two metrics. In particular for Skinner-C, the ratio of expected to optimal execution time is for all queries upper-bounded by a low-order polynomial in the query size. Given misleading statistics or assumptions, traditional query optimizers may select plans whose execution time is higher than optimal by a factor that is exponential in the number of tables joined. The same applies to adaptive processing strategies~\cite{Tzoumas2008} which, even if they converge to optimal join orders over time, do not bound the overhead caused by single tuples processed along bad join paths.

SkinnerDB pays for reliable join ordering with overheads for learning and join order switching. In our experiments with various baselines and benchmarks, we study under which circumstances the benefits outweigh the drawbacks. When considering accumulated execution time on difficult benchmarks (e.g., the join order benchmark~\cite{Gubichev2015}), it turns out that SkinnerDB can beat even highly performance-optimized systems for analytical processing with a traditional optimizer. While per-tuple processing overheads are significantly lower for the latter, SkinnerDB minimizes the total number of tuples processed via better join orders.

We summarize our original scientific contributions:

\begin{itemize}
\item We introduce a new quality criterion for query evaluation strategies that compares expected and optimal execution cost.
\item We propose several adaptive execution strategies based on reinforcement learning. 
\item We formally prove correctness and regret bounds for those execution strategies.
\item We experimentally compare those strategies, implemented in SkinnerDB, against various baselines.
\end{itemize}

The remainder of this paper is organized as follows. We discuss related work in Section~\ref{relatedSec}. We describe the primary components of SkinnerDB in Section~\ref{overviewSec}. In Section~\ref{algorithmSec}, we describe our query evaluation strategies based on reinforcement learning. In Section~\ref{analysisSec}, we analyze those strategies formally, we prove correctness and performance properties. Finally, in Section~\ref{experimentsSec}, we describe the implementation in SkinnerDB and compare our approaches experimentally against a diverse set of baselines. The appendix contains additional experimental results.

%% file: sectionsPdf/related.tex
\section{Related Work}
\label{relatedSec}

Our approach connects to prior work collecting information on predicate selectivity by evaluating them on data samples~\cite{Bruno2002, Chaudhuri2001, Haas1992, Haas2011, Karanasos2014a, Lipton1990a, Markl2013, Wu2016}. We compare in our experiments against a recently proposed representative~\cite{Wu2016}. Most prior approaches rely on a traditional optimizer to select interesting intermediate results to sample. They suffer if the original optimizer generates bad plans. The same applies to approaches for interleaved query execution and optimization~\cite{Aboulnaga2004a, Avnur2000, Babu2005} that repair initial plans at run time if cardinality estimates turn out to be wrong. Robust query optimization~\cite{Alyoubi2015, Alyoubi2016, Babcock2005, D.2008} assumes that predicate selectivity is known within narrow intervals which is often not the case~\cite{El-Helw2009}. Prior work~\cite{Dutt2014a, Dutt2014} on query optimization without selectivity estimation is based on simplifying assumptions (e.g., independent predicates) that are often violated.


Machine learning has been used to estimate cost for query plans whose cardinality values are known~\cite{Akdere2011, Li2012}, to predict query~\cite{Ganapathi} or workflow~\cite{Popescu2013} execution times, result cardinality~\cite{Malik2006, Malik2007}, or interference between query executions~\cite{Duggan2011}. LEO~\cite{Aboulnaga2004a, Stillger2001}, IBM's learning optimizer, leverages past query executions to improve cardinality estimates for similar queries. Ewen et al.~\cite{Ewen2005} use a similar approach for federated database systems. Several recent approaches~\cite{Krishnan2018, Marcus2018} use learning for join ordering. All of the aforementioned approaches learn from past queries for the optimization of future queries. To be effective, new queries must be similar to prior queries and this similarity must be recognizable. Instead, we learn \textit{during} the execution of a query. 




Adaptive processing strategies have been explored in prior work~\cite{Avnur2000, Deshpande2004, Deshpande2006a, Quanzhong2007a, Raman2003, Tzoumas2008, Viglas2003}. Our work uses reinforcement learning and is therefore most related to prior work using reinforcement learning in the context of Eddies~\cite{Tzoumas2008}. We compare against this approach in our experiments. Eddies do not provide formal guarantees on the relationship between expected execution time and the optimum. They never discard intermediate results, even if joining them with the remaining tables creates disproportional overheads. Eddies support bushy query plans in contrast to our approach. Bushy plans can in principle decrease execution cost compared to the best left-deep plan. However, optimal left-deep plans typically achieve reasonable performance~\cite{Gubichev2015}. Also, as we show in our experiments, reliably identifying near-optimal left-deep plans can be better than selecting bushy query plans via non-robust optimization.


Our work relates to prior work on filter ordering with regret bounds~\cite{Condon2009a}. Join ordering introduces however new challenges, compared to filter ordering. In particular, applying more filters can only decrease the size of intermediate results. The relative overhead of a bad filter order, compared to the optimum, grows therefore linearly in the number of filters. The overhead of bad join orders, compared to the optimum, can grow exponentially in the query size. This motivates mechanisms that bound join overheads for single data batches, as well as mechanisms to save progress for partially processed data batches. 

Worst-case optimal join algorithms~\cite{Ngo2012, Veldhuizen2012} bound cost as a function of worst-case query result size. We bound expected execution cost as a function of cost for processing an optimal join order. Further, prior work on worst-case optimal joins focuses on conjunctive queries while we support a broader class of queries, including queries with user-defined function predicates. Our approach applies to SQL with standard semantics while systems for worst-case optimal evaluation typically assume set semantics~\cite{Veldhuizen2012}. 

%% file: sectionsPdf/overview.tex
\section{Overview}
\label{overviewSec}

\tikzstyle{SkinnerComponent}=[anchor=center, draw=black, fill=blue!15, minimum width=1.6cm, align=center, font=\small]

\tikzstyle{SkinnerComponent2}=[anchor=center, draw=black, rounded corners=0.15cm, fill=red!15, minimum width=1.75cm, align=center, font=\small]

\tikzstyle{SkinnerIO}=[anchor=center, font=\small]

\tikzstyle{SkinnerFlow}=[ultra thick, draw]

\begin{figure}
\includegraphics{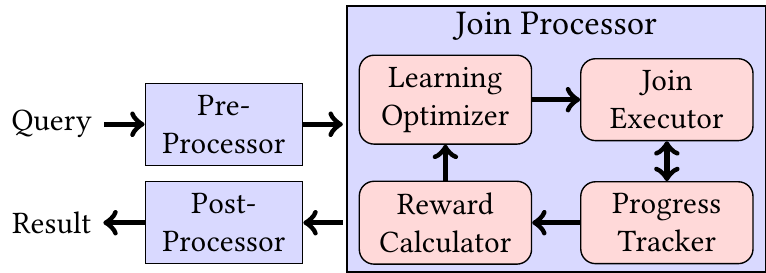}
\caption{Primary components of SkinnerDB.\label{architectureFig}}
\end{figure}

Figure~\ref{architectureFig} shows the primary components of SkinnerDB. This high-level outline applies to all of the SkinnerDB variants. 

The pre-processor is invoked first for each query. Here, we filter base tables via unary predicates. Also, depending on the SkinnerDB variant, we partition the remaining tuples into batches or hash them (to support joins with equality predicates). 

Join execution proceeds in small time slices. The join processor consists of several sub-components. The learning optimizer selects a join order to try next at the beginning of each time slice. It uses statistics on the quality of join orders that were collected during the current query execution. Selected join orders are forwarded to the join executor. This component executes the join order until a small timeout is reached. We add result tuples into a result set, checking for duplicate results generated by different join orders. The join executor can be either a generic SQL processor or, for maximal performance, a specialized execution engine. The same join order may get selected repeatedly. The progress tracker keeps track of which input data has been processed already. For Skinner-C, it even tracks execution state for each join order tried so far, and merges progress across join orders. At the start of each time slice, we consult the progress tracker to restore the latest state stored for the current join order. At the end of it, we backup progress achieved during the current time slice. The reward calculator calculates a reward value, based on progress achieved during the current time slice. This reward is a measure for how quickly execution proceeds using the chosen join order. It is used as input by the optimizer to determine the most interesting join order to try in the next time slice.

Finally, we invoke the post-processor, using the join result tuples as input. Post-processing involves grouping, aggregation, and sorting. In the next section, we describe the algorithms executed within SkinnerDB.




%% file: sectionsPdf/algorithm.tex
\section{Algorithms}
\label{algorithmSec}

We describe several adaptive processing strategies that are implemented in SkinnerDB. In Section~\ref{uctSub}, we introduce the UCT algorithm that all processing strategies are based upon. In Section~\ref{learningSub}, we describe how the UCT algorithm can generally be used to learn optimal join orders. In Section~\ref{genericSub}, we introduce a join order learning approach that can be implemented on top of existing SQL processing engines, in a completely non-intrusive manner. In Section~\ref{hybridSub}, we show how this strategy can integrate plans proposed by the original optimizer. In Section~\ref{customizedSub}, we propose a new query evaluation method that facilitates join order learning and the associated learning strategy. 

While we describe the following algorithms only for SPJ queries, it is straight-forward to add sorting, grouping, or aggregate calculations in a post-processing step (we do so in our actual implementation). Nested queries can be treated via decomposition~\cite{Neumann}.

\subsection{Background on UCT}
\label{uctSub}

Our method for learning optimal join orders is based on the UCT algorithm~\cite{Kocsis2006}. This is an algorithm from the area of reinforcement learning. It assumes the following scenario. We repeatedly make choices that result in rewards. Each choice is associated with reward probabilities that we can learn over time. Our goal is to maximize the sum of obtained rewards. To achieve that goal, it can be beneficial to make choices that  resulted in large rewards in the past (``exploitation'') or choices about which we have little information (``exploration'') to inform future choices. The UCT algorithm balances between exploration and exploitation in a principled manner that results in probabilistic guarantees. More precisely, assuming that rewards are drawn from the interval $[0,1]$, the UCT algorithm guarantees that the expected regret (i.e., the difference between the sum of obtained rewards to the sum of rewards for optimal choices) is in $O(\log(n))$ where $n$ designates the number of choices made~\cite{Kocsis2006}. 

We specifically select the UCT algorithm for several reasons. First, UCT has been applied successfully to problems with very large search spaces (e.g., planning Go moves~\cite{Gelly2012}). This is important since the search space for join ordering grows quickly in the query size. Second, UCT provides formal guarantees on cumulative regret (i.e., accumulated regret over all choices made). Other algorithms from the area of reinforcement learning~\cite{Feldman2014} focus for instance on minimizing simple regret (i.e., quality of the final choice). The latter would be more appropriate when separating planning from execution. Our goal is to interleave planning and execution, making the first metric more appropriate. Third, the formal guarantees of UCT do not depend on any instance-specific parameter settings~\cite{Domshlak2013}, distinguishing it from other reinforcement learning algorithms.

We assume that the space of choices can be represented as a search tree. In each round, the UCT algorithm makes a series of decisions that can be represented as a path from the tree root to a leaf. Those decisions result in a reward from the interval $[0,1]$, calculated by an arbitrary, randomized function specific to the leaf node (or as a sum of rewards associated with each path step). Typically, the UCT algorithm is applied in scenarios where materializing the entire tree (in memory) is prohibitively expensive. Instead, the UCT algorithm expands a partial search tree gradually towards promising parts of the search space. The UCT variant used in our system expands the materialized search tree by at most one node per round (adding the first node on the current path that is outside the currently materialized tree). 

Materializing search tree nodes allows to associate statistics with each node. The UCT algorithm maintains two counters per node: the number of times the node was visited and the average reward that was obtained for paths crossing through that node. If counters are available for all relevant nodes, the UCT algorithm selects at each step the child node $c$ maximizing the formula $r_c+w\cdot \sqrt{\log(v_p)/v_c}$ where $r_c$ is the average reward for $c$, $v_c$ and $v_p$ are the number of visits for child and parent node, and $w$ a weight factor. In this formula, the first term represents exploitation while the second term represents exploration. Their sum represents the upper bound of a confidence bound on the reward achievable by passing through the corresponding node (hence the name of the algorithm: UCT for Upper Confidence bounds applied to Trees). Setting $w=\sqrt{2}$ is sufficient to obtain bounds on expected regret. It can however be beneficial to try different values to optimize performance for specific domains~\cite{Domshlak2013}.

\subsection{Learning Optimal Join Orders}
\label{learningSub}


Our search space is the space of join orders. We consider all join orders except for join orders that introduce Cartesian product joins without need. Avoiding Cartesian product joins is a very common heuristic that is used by virtually all optimizers~\cite{Gubichev2015}. 
To apply the UCT algorithm for join ordering, we need to represent the search space as a tree. We assume that each tree node represents one decision with regards to the next table in the join order. Tree edges represent the choice of one specific table. The tree root represents the choice of the first table in the join order. All query tables can be chosen since no table has been selected previously. Hence, the root node will have $n$ child nodes where $n$ is the number of tables to join. Nodes in the next layer of the tree (directly below the root) represent the choice of a second table. We cannot select the same table twice in the same join order. Hence, each of the latter node will have at most $n-1$ child nodes associated with remaining choices. The number of choices depends on the structure of the join graph. If at least one of the remaining tables is connected to the first table via join predicates, only such tables will be considered. If none of the remaining tables is connected, all remaining tables become eligible (since a Cartesian product join cannot be avoided given the initial choice). In total, the search tree will have $n$ levels. Each leaf node is associated with a completely specified join order. 

We generally divide the execution of a query into small time slices in which different join order are tried. For each time slice, the UCT algorithm selects a path through the aforementioned tree, thereby selecting the join order to try next. As discussed previously, only part of the tree will be ``materialized'' (i.e., we keep nodes with node-specific counters in main memory). When selecting a path (i.e., a join order), UCT exploits counters in materialized nodes wherever available to select the next path step. Otherwise, the next step is selected randomly. After a join order has been selected, this join order is executed during the current time slice. Results from different time slices are merged (while removing overlapping results). We stop once a complete query result is obtained.

Our goal is to translate the aforementioned formal guarantees of UCT, bounding the distance between expected and optimal reward (i.e., the regret), into guarantees on query evaluation speed. To achieve that goal, we must link the reward function to query evaluation progress. The approaches for combined join order learning and execution, presented in the following subsections, define the reward function in different ways. They all have however the property that higher rewards correlate with better join orders. After executing the selected join order for a bounded amount of time, we measure evaluation progress and calculate a corresponding reward value. The UCT algorithm updates counters (average reward and number of visits) in all materialized tree nodes on the previously selected path. 

The following algorithms use the UCT algorithm as a sub-function.  More precisely, we use two UCT-related commands in the following pseudo-code: \Call{UctChoice}{$T$} and \Call{RewardUpdate}{$T,j,r$}. The first one returns the join order chosen by the UCT algorithm when applied to search tree $T$ (some of the following processing strategies maintain multiple UCT search trees for the same query). The second function updates tree $T$ by registering reward $r$ for join order $j$. Sometimes, we will pass a reward function instead of a constant for $r$ (with the semantics that the reward resulting from an evaluation of that function is registered).

\subsection{Generic Execution Engines}
\label{genericSub}

In this subsection, we show how we can learn optimal join orders when treating the execution engine as a black box with an SQL interface. This approach can be used on top of existing DBMS without changing a single line of their code. 

A naive approach to learn optimal join orders in this context would be the following. Following the discussion in the last subsection, we divide each table joined by the input query into an equal number of batches (if the input query contains unary predicates in the where clause, we can apply them in the same step). We simplify by assuming that all tables are sufficiently large to contain at least one tuple per batch (otherwise, less batches can be used for extremely small tables). We iteratively choose join orders using the UCT algorithm. In each iteration, we use the given join order to process a join between one batch \textit{for the left most table in the join order} and the remaining, complete tables. We remove each processed batch and add the result of each iteration to a result relation. We terminate processing once all batches are processed for at least one table. As we prove in more detail in Section~\ref{analysisSec}, the result relation contains a complete query result at this point. To process the query as quickly as possible, we feed the UCT algorithm with a reward function that is based on processing time for the current iteration. The lower execution time, the higher the corresponding reward. Note that reducing the size of the left-most table in a join order (by using only a single batch) tends to reduce the sizes of all intermediate results. If the dominant execution time component is proportional to those intermediate result sizes (e.g., time for generating intermediate result tuples, index lookups, number of evaluated join predicates), execution time for one batch is proportional to execution time for the entire table (with a scaling factor that corresponds to the number of batches per table). 

The reason why we call the latter algorithm naive is the following. In many settings, the reward function for the UCT algorithm is relatively inexpensive to evaluate. In our case, it requires executing a join between one batch and all the remaining tables. The problem is that execution cost can vary strongly as a function of join order. The factor separating execution time of best and worst join order may grow exponentially in the number of query tables. Hence, even a single iteration with a bad join order and a single tuple in the left-most table may lead to an overall execution time that is far from the optimum for the entire query. Hence, we must upper-bound execution time in each iteration.


This leads however to a new problem: what timeout should we choose per batch in each iteration? Ideally, we would select as timeout the time required by an optimal join order. Of course, we neither know an optimal join order nor its optimal processing time for a new query. Using a timeout that is lower than the optimum prevents us from processing an entire batch before the timeout. This might be less critical if we can backup the state of the processing engine and restore it when trying the same join order again. However, we currently treat the processing engine as a black box and cannot assume access to partial results and internal state. Further, most SQL processing engines execute a series of binary joins and generate potentially large intermediate results. As we may try out many different join orders, already the space required for storing intermediate results for each join order would become prohibitive. So, we must assume that all intermediate results are lost if execution times out before a batch is finished. Using lower timeouts than necessary prevents us from making any progress. On the other side, choosing a timeout that is too high leads to unnecessary overheads when processing sub-optimal join orders. 

\begin{figure}[t]
\centering
\includegraphics{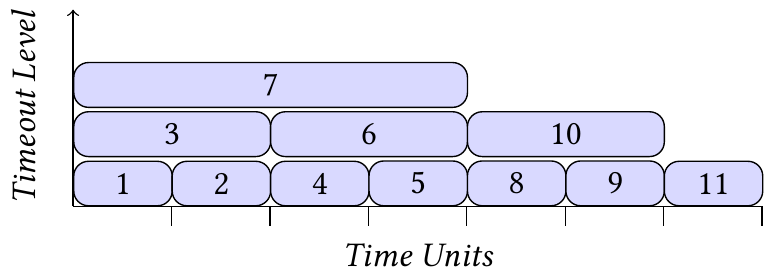}
\caption{Illustration of time budget allocation scheme: we do not know the optimal time per batch and iterate over different timeouts, allocating higher budgets less frequently.\label{budgetFigure}}
\end{figure}

The choice of a good timeout is therefore crucial while we cannot know the best timeout a-priori. The solution lies in an iterative scheme that tries different timeouts in different iterations. We carefully balance allocated execution time over different timeouts, avoiding to use higher timeouts unless lower ones have been tried sufficiently often. More precisely, we will present a timeout scheme that ensures that the total execution time allocated per timeout does not differ by more than factor two across different timeouts. Figure~\ref{budgetFigure} gives an intuition for the corresponding timeout scheme (numbers indicate the iteration in which the corresponding timeout is chosen). We use timeouts that are powers of two (we also call the exponent the \textit{Level} of the timeout). We always choose the highest timeout for the next iteration such that the accumulated execution time for that timeout does not exceed time allocated to any lower timeout. Having fixed a timeout for each iteration, we assign a reward of one for a fixed join order if the input was processed entirely. We assign a reward of zero otherwise. 

Algorithm~\ref{nonIntrusiveAlg} present pseudo-code matching the verbal description. First, tuples are filtered using unary predicates and the remaining tuples are partitioned into $b$ batches per table (we omit pseudo-code for pre-processing). We use function~\Call{DBMS}{} to invoke the underlying DBMS for processing one batch with a timeout. The function accumulates partial result in a result relation if processing finishes before the timeout and returns \textbf{true} in that case. Vector $o_i$ stores for each table an offset, indicating how many of its batches were completely processed (it is implicitly initialized to one for each table). Variable $n_l$ stores for each timeout level $l$ how much execution time was dedicated to it so far (it is implicitly initialized to zero and updated in each invocation of function~\Call{NextTimeout}{}). Note that we maintain separate UCT trees $T_t$ for each timeout $t$ (implicitly initialized as a single root node representing no joined tables). This prevents for instance processing failures for lower timeouts to influence join ordering decisions for larger timeouts. We prove the postulated properties of the timeout scheme (i.e., balancing time over different timeouts) in Section~\ref{analysisSec}. 

\begin{algorithm}[t!]
\renewcommand{\algorithmiccomment}[1]{// #1}
\begin{small}
\begin{algorithmic}[1]
\State \Comment{Returns timeout for processing next batch,}
\State \Comment{based on times $n$ given to each timeout before.}
\Function{NextTimeout}{$n$}
\State \Comment{Choose timeout level}
\State $L\gets\max\{L|\forall l<L:n_l\geq n_L+2^L\}$
\State \Comment{Update total time given to level}
\State $n_L\gets n_L+2^L$
\State \Comment{Return timeout for chosen level}
\State \Return{$2^L$}
\EndFunction
\vspace{0.15cm}
\State \Comment{Process SPJ query $q$ using existing DBMS and}
\State \Comment{by dividing each table into $b$ batches.}
\Procedure{SkinnerG}{$q=R_1\Join\ldots\Join R_m,b$}
\State \Comment{Apply unary predicates and partitioning}
\State $\{R_1^1,\ldots,R_m^b\}\gets$\Call{PreprocessingG}{$q,b$}
\State \Comment{Until we processed all batches of one table}
\While{$\nexists i:o_i>b$}
\State \Comment{Select timeout using pyramid scheme}
\State $t\gets$\Call{NextTimeout}{n}
\State \Comment{Select join order via UCT algorithm}
\State $j\gets$\Call{UctChoice}{$T_t$}
\State \Comment{Process one batch until timeout}
\State $suc\gets$\Call{DBMS}{$R_{j1}^{o_{j1}}\Join R_{j2}^{o_{j2}..b}\ldots\Join R_{jm}^{o_{jm}..b},t$}
\State \Comment{Was entire batch processed successfully?}
\If{$suc$}
\State \Comment{Mark current batch as processed}
\State $o_{j1}\gets o_{j1}+1$
\State \Comment{Store maximal reward in search tree}
\State \Call{RewardUpdate}{$T_t,j,1$}
\Else
\State \Comment{Store minimal reward in search tree}
\State \Call{RewardUpdate}{$T_t,j,0$}
\EndIf
\EndWhile
\EndProcedure
\end{algorithmic}
\end{small}
\caption{Regret-bounded query evaluation using a generic execution engine.\label{nonIntrusiveAlg}}
\end{algorithm}

\subsection{Hybrid Algorithm}
\label{hybridSub}


The algorithm presented in the last subsection uses reinforcement learning alone to order joins. It bypasses any join ordering capabilities offered by an existing optimizer completely. This approach is efficient for queries where erroneous statistics or difficult-to-analyze predicates mislead the traditional optimizer. However, it adds unnecessary learning overheads for standard queries where a traditional optimizer would produce reasonable query plans.

\begin{figure}[t]
\includegraphics{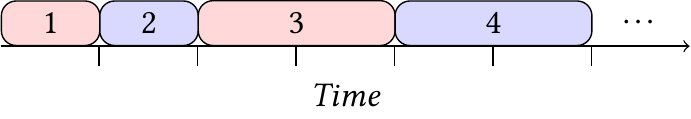}
\caption{The hybrid approach alternates with increasing timeouts between executing plans proposed by the traditional optimizer (red) and learned plans (blue).\label{hybridFig}}
\end{figure}

We present a hybrid algorithm that combines reinforcement learning with a traditional query optimizer. Instead of using an existing DBMS only as an execution engine, we additionally try benefiting from its query optimizer whenever possible. We do not provide pseudo-code for the hybrid algorithm as it is quick to explain. We iteratively execute the query using the plan chosen by the traditional query optimizer, using a timeout of $2^i$ where $i$ is the number of invocations (for the same input query) and time is measured according to some atomic units (e.g., several tens of milliseconds). In between two traditional optimizer invocations, we execute the learning based algorithm described in the last subsection. We execute it for the same amount of time as the traditional optimizer. We save the state of the UCT search trees between different invocations of the learning approach. Optionally, if a table batch was processed by the latter, we can remove the corresponding tuples before invoking the traditional optimizer. Figure~\ref{hybridFig} illustrates the hybrid approach. As shown in Section~\ref{analysisSec}, the hybrid approach bounds expected regret (compared to the optimal plan) and guarantees a constant factor overhead compared to the original optimizer.



\subsection{Customized Execution Engines}
\label{customizedSub}

The algorithms presented in the previous sections can work with any execution engine for SPJ queries. In this section, we present an execution engine that is tailored towards the needs of a learning based join ordering strategy.  In addition, we present a variant of the join order learning algorithm that optimally exploits that execution engine.

\begin{algorithm}[t]
\renewcommand{\algorithmiccomment}[1]{// #1}
\begin{small}
\begin{algorithmic}[1]
\State \Comment{Advance tuple index in state $s$ for table at position $i$}
\State \Comment{in join order $j$ for query $q$, considering tuple offsets $o$.}
\Function{NextTuple}{$q=R_1\Join\ldots\Join R_m,j,o,s,i$}
\State \Comment{Advance tuple index for join order position}
\State $s_{j_i}\gets s_{j_i}+1$
\State \Comment{While index exceeds relation cardinality}
\While{$s_{j_i}>|R_{j_i}|$ \textbf{and} $i>0$}
\State $s_{j_i}\gets o_{j_i}$
\State $i\gets i-1$
\State $s_{j_i}\gets s_{j_i}+1$
\EndWhile
\State \Return{$\langle s,i\rangle$}
\EndFunction
\vspace{0.15cm}
\State \Comment{Execute join order $j$ for query $q$ starting from}
\State \Comment{tuple indices $s$ with tuple offsets $o$. Add results}
\State \Comment{to $R$ until time budget $b$ is depleted.}
\Function{ContinueJoin}{$q=R_1\Join\ldots\Join R_m,j,o,b,s,R$}
\State $i\gets1$ \Comment{Initialize join order index}
\While{processing time $<b$ \textbf{and} $i>0$}
\State $t\gets\Call{Materialize}{R_{j_1}[s_{j_1}]\times\ldots\times R_{j_i}[s_{j_i}]}$
\If{$t$ satisfies all newly applicable predicates}
\If{$i=m$} \Comment{Is result tuple completed?}
\State $R\gets R\cup\{s\}$ \Comment{Add indices to result set}
\State $\langle s,i\rangle\gets\Call{NextTuple}{q,j,o,s,i}$ 
\Else \Comment{Tuple is incomplete}
\State $i\gets i+1$
\EndIf
\Else \Comment{Tuple violates predicates}
\State $\langle s,i\rangle\gets\Call{NextTuple}{q,j,o,s,i}$ 
\EndIf
\EndWhile
\State \Comment{Join order position 0 indicates termination}
\State \Return{$(i<1)$}
\EndFunction
\end{algorithmic}
\end{small}
\caption{Multi-way join algorithm supporting fast join order switching.\label{customizedAuxAlg}}
\end{algorithm}

\begin{algorithm}[t]
\renewcommand{\algorithmiccomment}[1]{// #1}
\begin{small}
\begin{algorithmic}[1]
\State \Comment{Regret-bounded evaluation of SPJ query $q$,}
\State \Comment{length of time slices is restricted by $b$.}
\Function{SkinnerC}{$q=R_1\Join\ldots\Join R_m,b$}
\State \Comment{Apply unary predicates and hashing}
\State $q\gets$\Call{PreprocessingC}{$q$}
\State $R\gets\emptyset$ \Comment{Initialize result indices}
\State $finished\gets\mathbf{false}$ \Comment{Initialize termination flag}
\While{$\neg finished$}
\State \Comment{Choose join order via UCT algorithm}
\State $j\gets\Call{UctChoice}{T}$
\State \Comment{Restore execution state for this join order}
\State $s\gets\Call{RestoreState}{j,o,S}; s_{prior}\gets s$
\State \Comment{Execute join order during time budget}
\State $finished\gets\Call{ContinueJoin}{q,j,o,b,s,R}$
\State \Comment{Update UCT tree via progress-based rewards}
\State \Call{RewardUpdate}{$T,j,\textproc{Reward}(s-s_{prior},j)$}
\State \Comment{Backup execution state for join order}
\State $\langle o,S\rangle\gets\Call{BackupState}{j,s,o,S}$
\EndWhile
\State \Return{$[\Call{Materialize}{R_1[s_1]\times R_2[s_2]\ldots}|s\in R]$}
\EndFunction
\end{algorithmic}
\end{small}
\caption{Regret-bounded query evaluation using a customized execution engine.\label{customizedAlg}}
\end{algorithm}

Most execution engines are designed for a traditional approach to query evaluation. They assume that a single join order is executed for a given query (after being generated by the optimizer). Learning optimal join orders while executing a query leads to unique requirements on the execution engine. First, we execute many different join orders for the same query, each one only for a short amount of time. Second, we may even execute the same join order multiple times with many interruptions (during which we try different join orders). This specific scenario leads to (at least) three desirable performance properties for the execution engine. First, the execution engine should minimize overheads when switching join orders. Second, the engine should preserve progress achieved for a given join order even if execution is interrupted. Finally, the engine should allow to share achieved progress, to the maximal extent possible, between different join orders as well. The generic approach realizes the latter point only to a limited extend (by discarding batches processed completely by any join order from consideration by other join orders). 

The key towards achieving the first two desiderata (i.e., minimal overhead when switching join orders or interrupting execution) is a mechanism that backs up execution state as completely as possible. Also, restoring prior state when switching join order must be very efficient. By ``state'', we mean the sum of all intermediate results and changes to auxiliary data structures that were achieved during a partial query evaluation for one specific join order. We must keep execution state as small as possible in order to back it up and to restore it efficiently. 


Two key ideas enable us to keep execution state small. First, we represent tuples in intermediate results concisely as vectors of tuple indices (each index pointing to one tuple in a base table). Second, we use a multi-way join strategy limiting the number of intermediate result tuples to at most one at any point in time. Next, we discuss both ideas in detail. 


Traditional execution engines for SPJ queries produce intermediate results that consist of actual tuples (potentially containing many columns with elevated byte sizes). To reduce the size of the execution state, we materialize tuples only on demand. Each tuple, be it a result tuple or a tuple in an intermediate result, is the result of a join between single tuples in a subset of base tables. Hence, whenever possible, we describe tuples simply by an array of tuple indices (whose length is bounded by the number of tables in the input query). We materialize partial tuples (i.e., only the required columns) temporarily to check whether they satisfy applicable predicates or immediately before returning results to the user. To do that efficiently, we assume a column store architecture (allowing quick access to selected columns) and a main-memory resident data set (reducing the penalty of random data access). 

Most traditional execution engines for SPJ queries process join orders by a sequence of binary join operations. This can generate large intermediate results that would become part of the execution state. We avoid that by a multi-way join strategy whose intermediate result size is restricted to at most one tuple. We describe this strategy first for queries with generic predicates. Later, we discuss an extension for queries with equality join predicates based on hashing.

Intuitively, our multi-way join strategy can be understood as a depth-first search for result tuples. Considering input tables in one specific join order, we fix one tuple in a predecessor table before considering tuples in the successor table. We start with the first tuple in the first table (in join order). Next, we select the first tuple in the second table and verify whether all applicable predicates are satisfied. If that is the case, we proceed to considering tuples in the third table. If not, we consider the next tuple in the second table. Once all tuples in the second table have been considered for a fixed tuple in the first table, we ``backtrack'' and advance the tuple indices for the first table by one. Execution ends once all tuples in the first table have been considered.

\tikzstyle{row}=[anchor=south west, draw=black, fill=blue!50, minimum width=1cm, minimum height=0.5cm]
\tikzstyle{joinFlow}=[thick, -stealth, red, dashed]
\tikzstyle{joinOrder}=[red, font=\small]

\begin{figure}[t]
\includegraphics{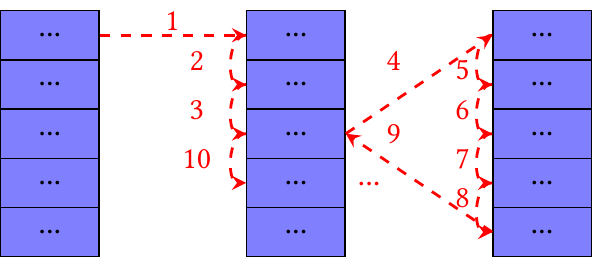}
\caption{Depth-first multi-way join strategy: we increase the join order index once the first tuple satisfying all applicable predicates is found, we decrease it once all tuples in the current table were considered.\label{joinFigure}}
\end{figure}

\begin{example}
Figure~\ref{joinFigure} illustrates the process for a three-table join. Having fixed a tuple in the left-most table (at the left, we start with the first tuple), the join order index is increased. Next, we find the first tuple in the second table satisfying the join condition with the current tuple in the first table. Having found such a tuple, we increase the join order index again. Now, we iterate over tuples in the third table, adding each tuple combination satisfying all applicable conditions to the result. After all tuples in the last table have been considered, we decrease the join order index and consider the next tuple in the second table.
\end{example}

Algorithm~\ref{customizedAuxAlg} implements that approach. Function~\Call{ContinueJoin}{} realizes the execution strategy described before. For a fixed amount of processing time (we use the number of outer while loop iterations as a proxy in our implementation) or until all input data is processed, it either increases ``depth'' (i.e., join order index $i$) to complete a partial tuple, satisfying all applicable predicates, further, or it advances tuples indices using Function~\Call{NextTuple}{}. The latter function increases the tuple indices for the current join order index or backtracks if the table cardinality is exceeded. Note that the same result tuple might be added multiple times in invocations of the execution engine for different join orders. However, we add tuple index vectors into a result \textit{set}, avoiding duplicate entries (of course, two different tuple index vectors can represent two result tuples with the same values in each column). 


We discuss the main function (\Call{SkinnerC}{}) learning optimal join orders using a customized execution engine (see Algorithm~\ref{customizedAlg}). The most apparent difference to the version from Section~\ref{genericSub} is the lack of a dynamic timeout scheme. Instead, we use the same timeout for each invocation of the execution engine. This becomes possible since progress made when executing a specific join order is never lost. By minimizing the size of the execution state, we have enabled an efficient backup and restore mechanism (encapsulated by functions \Call{BackupState}{} and \Call{RestoreState}{} whose pseudo-code we omit) that operates only on a small vector of indices. The number of stored vectors is furthermore proportional to the size of the UCT tree. The fact that we do not lose partial results due to inappropriate timeouts anymore has huge impact from the theoretical perspective (see Section~\ref{analysisSec}) as well as for performance in practice (see Section~\ref{experimentsSec}). Learning overheads are lower than before since we only maintain a single UCT search tree accumulating knowledge from all executions. 



In Section~\ref{genericSub}, we used a binary reward function based on whether the current batch was processed. We do not process data batch-wise anymore and must therefore change the reward function (represented as function~\textproc{Reward} in the pseudo-code which depends on execution state delta and join order). For instance, we can we use as reward the percentage of tuples processed in the left-most table during the last invocation. This function correlates with execution speed and returns values in the range between 0 and 1 (the standard formulas used for selecting actions by the UCT algorithm are optimized for that case~\cite{Kocsis2006}). SkinnerDB uses a slight refinement: we sum over all tuple index deltas, scaling each one down by the product of cardinality values of its associated table and the preceding tables in the current join order. Note that the UCT algorithm averages rewards over multiple invocations of the same join order and keeps exploring (i.e., obtaining a reward of zero for one good join order during a single invocation of the execution engine will not exclude that order from further consideration). 

We have not yet discussed how our approach satisfies the third desiderata (sharing as much progress as possible among different join orders) mentioned at the beginning. We use in fact several techniques to share progress between different join orders (those techniques are encapsulated in Function~\Call{RestoreState}{}). First, we use again offset counters to exclude for each table tuples that have been joined with all other tuples already (vector $o$ in the pseudo-code which is implicitly initialized to one). In contrast to the version from Section~\ref{genericSub}, offsets are not defined at the granularity of data batches but at the granularity of single tuples. This allows for a more fine-grained sharing of progress between different join orders than before. 

Second, we share progress between all join orders with the same prefix. Whenever we restore state for a given join order, we compare execution progress between the current join order and all other orders with the same prefix (iterating over all possible prefix lengths). Comparing execution states $s$ and $s'$ for two join orders $j$ and $j'$ with the same prefix of length $k$ (i.e., the first $k$ tables are identical), the first order is ``ahead'' of the second if there is a join order position $p\leq k$ such that $s_{j_i}\geq s'_{j_i}$ for $i<p$ and $s_{j_p}>s'_{j_p}+1$. In that case, we can ``fast-forward'' execution of the second join order, skipping result tuples that were already generated via the first join order. We do so by executing $j'$ from a merged state $s''$ where $s''_{j'_i}=s_{j'_i}$ for $i<p$, $s''_{j'_p}=s_{j'_p}-1$, and $s''_{j'_i}=o_{j'_i}$ for $i>p$ (since we can only share progress for the common prefix). Progress for different join orders is stored in the data structure represented as $S$ in Algorithm~\ref{customizedAlg}, Function~\textproc{RestoreState} takes care of fast-forwarding (selecting the most advanced execution state among all alternatives).


So far, we described the algorithm for queries with generic predicates. Our actual implementation uses an extended version supporting equality join predicates via hashing. If equality join predicates are present, we create hash tables on all columns subject to equality predicates during pre-processing. Of course, creating hash tables to support all possible join orders creates overheads. However, those overheads are typically small as only tuples satisfying all unary predicates are hashed. We extend Algorithm~\ref{customizedAuxAlg} to benefit from hash tables: instead of incrementing tuple indices always by one (line~5), we ``jump'' directly to the next highest tuple index that satisfies at least all applicable equality predicates with preceding tables in the current join order (this index can be determined efficiently via probing).





%% file: sectionsPdf/analysis.tex
\section{Formal Analysis}
\label{analysisSec}

We prove correctness (see Section~\ref{correctnessSub}), and the regret bounds (see Section~\ref{regretSub}) for all Skinner variants.

\subsection{Correctness}
\label{correctnessSub}

Next, we prove correctness (i.e., that each algorithm produces a correct query result). We distinguish result tuples (tuples from the result relation joining all query tables) from component tuples (tuples taken from a single table). 

\begin{theorem}
Skinner-G produces the correct query result.
\end{theorem}
\begin{proof}
Offsets exclude component tuples from consideration when executing the following joins. We show the following invariant: all result tuples containing excluded component tuples have been generated. This is certainly true at the start where offsets do not exclude any tuples. Offsets are only advanced if batches have been successfully processed. In that case, all newly excluded component tuples have been joined with tuples from all other tables that are not excluded. But excluded tuples can be neglected according to our invariant. The algorithm terminates only after all tuples from one table have been excluded. In that case, all result tuples have been generated. Still, we need to show that no result tuple has been generated more often than with a traditional execution. This is the case since we exclude all component tuples in one table after each successfully processed batch.
\end{proof}

\begin{theorem}
Skinner-H produces the correct query result.
\end{theorem}
\begin{proof}
We assume that executing a query plan produced by the traditional optimizer generates a correct result. The result produced by Skinner-G is correct according to the preceding theorem. This implies that Skinner-H produces a correct result as it returns the result generated by one of the latter two algorithms.
\end{proof}

\begin{theorem}
Skinner-C produces the correct query result.
\end{theorem}
\begin{proof}
Skinner-C does not produce any duplicate result tuples as justified next. Result tuples are materialized only at the very end of the main function. The result set contains tuple index vectors until then. Vectors are unique over all result tuples (as they indicate the component tuples from which they have been formed) and, due to set semantics, no vector will be contained twice in the result. Also, Skinner-C produces each result tuple at least once. This is due to the fact that \textit{i)}~complete tuples are always inserted into the result set, \textit{ii)}~partial tuples (i.e., $i<m$) are completed unless they violate predicates (then they cannot be completed into result tuples), and \textit{iii)}~tuple indices are advanced in a way that covers all combinations of component tuples.
\end{proof}

\subsection{Regret Bounds}
\label{regretSub}

Regret is the difference between actual and optimal execution time. We denote execution time by $n$ and optimal time by $n^*$. Skinner-G and Skinner-H choose timeout levels (represented by the $y$ axis in Figure~\ref{budgetFigure}) that we denote by $l$. We use the subscript notation (e.g., $n_l$) to denote accumulated execution time spent with a specific timeout level. We study regret for fixed query properties (e.g., the number of joined tables, $m$, or the optimal reward per time slice, $r^*$) for growing amounts of input data (i.e., table size) and execution time. In particular, we assume that execution time, in relation to query size, is large enough to make the impact of transitory regret negligible~\cite{Coquelin2007b}. We focus on regret of the join phase as pre-processing overheads are linear in data and query size (while post-processing overheads are polynomial in query and join result size). We assume that time slices are chosen large enough to make overheads related to learning and join order switching negligible. Specifically for Skinner-G and Skinner-H, we assume that the optimal timeout per time slice applies to all batches. To simplify the analysis, we study slightly simplified versions of the algorithms from Section~\ref{algorithmSec}. In particular, we assume that offsets are only applied to exclude tuples for the left-most table in the current join order. This means that no progress is shared between join orders that do not share the left-most table. For Skinner-C, we assume that the simpler reward function (progress in left-most table only) is used. We base our analysis on the properties of the UCT variant proposed by Kocsis and Szepesvari~\cite{Kocsis2006}.

For a given join order, processing time in SkinnerDB is equivalent to processing time in traditional engines if scaling down the size of the left-most table scales down execution time proportionally (i.e., execution time behaves similarly to the $C_{out}$ cost metric~\cite{Krishnamurthy1986}). If so, the regret bounds apply compared to an optimal traditional query plan execution.





Before analyzing Skinner-G, we first prove several properties of the pyramid timeout scheme introduced in Section~\ref{genericSub}. 

\begin{lemma}
The number of timeout levels used by Skinner-G is upper-bounded by $\log(n)$.\label{nrLevelsLemma}
\end{lemma}
\begin{proof}
We add a new timeout level $L$, whenever the equation $n_l\geq n_L+2^L$ is satisfied for all $0\leq l<L$ for the first time. As $n_l$ is generally a sum over powers of two ($2^l$), and as $n_L=0$ before $L$ is used for the first time, the latter condition can be tightened to $2^L=n_l$ for all $0\leq l<L$. Hence, we add a new timeout whenever the total execution time so far can be represented as $L\cdot 2^L$ for $L\in\mathbb{N}$. Assuming that $n$ is large, specifically $n>1$, the number of levels grows faster if adding levels whenever execution time can be represented as $2^L$ for $L\in\mathbb{N}$. In that case, the number of levels can be bounded by $\log(n)$ (using the binary logarithm). 
\end{proof}

\begin{lemma}\label{balancedLevelsLemma}
The total amount of execution time allocated to different (already used) timeout levels cannot differ by more than factor two.
\end{lemma}
\begin{proof}
Assume the allocated time differs by more than factor two between two timeout levels, i.e.\ $\exists l_1,l_2:n_{l_1}>2\cdot n_{l_2}$ (and $n_{l_1},n_{l_2}\neq 0$). Consider the situation in which this happens for the first time. Since $\forall i:n_i\geq n_{i+1}$, we must have $n_0>2\cdot n_L$ where $L$ is the largest timeout level used so far. This was not the case previously so we either selected timeout level 0 or a new timeout level $L$ in the last step. If we selected a new timeout level $L$ then it was $n_l\geq n_L+2^L$ for all $0\leq l<L$ which can be tightened to $\forall 0\leq l<L:n_l=2^L$ (exploiting that $n_L=0$ previously and that timeouts are powers of two). Hence, selecting a new timeout cannot increase the maximal ratio of time per level. Assume now that timeout level 0 was selected. Denote by $\delta_{i}=n_i-n_{i+1}$ for $i<L$ the difference in allocated execution time between consecutive levels, before the last selection. It is $\delta_{i}\leq 2^{i}$ since $n_{i}$ is increased in steps of size $2^{i}$ and strictly smaller than $2^{i+1}$ (otherwise, level $i+1$ or a higher one would have been selected). It was $n_0-n_L=\sum_{0\leq i<L}\delta_i\leq \sum_{0\leq i<L}2^i< 2^L$. On the other side, it was $n_L\geq 2^L$ (as $n_L\neq0$ and since $n_L$ is increased in steps of $2^L$). After $n_0$ is increased by one, it is still $n_0\leq 2\cdot n_L$. The initial assumption leads always to a contradiction.
\end{proof}

We are now ready to provide worst-case bounds on the expected regret when evaluating queries via Skinner-G.

\begin{theorem}\label{skinnerGtheorem}
Expected execution time regret of Skinner-G is upper-bounded by $(1-1/(\log(n)\cdot m\cdot 4))\cdot n+O(\log(n))$.
\end{theorem}
\begin{proof}
Total execution time $n$ is the sum over execution time components $n_l$ that we spent using timeout level $l$, i.e.\ we have $n=\sum_{0\leq l\leq L}n_l$ where $L+1$ is the number of timeout levels used. It is $L+1\leq \log(n)$ due to Lemma~\ref{nrLevelsLemma} and $\forall l_1,l_2\in L:n_{l_1}\geq n_{l_2}/2$ due to Lemma~\ref{balancedLevelsLemma}. Hence, for any specific timeout level $l$, we have $n_l\geq n/(2\cdot\log(n))$. Denote by $l^*$ the smallest timeout, tried by the pyramid timeout scheme, which allows to process an entire batch using the optimal join order. It is $n_{l^*}\geq n/(2\cdot\log(n))$. We also have $n_{l^*}=n_{l^*,1}+n_{l^*,0}$ where $n_{l^*,1}$ designates time spent executing join orders with timeout level $l^*$ that resulted in reward $1$, $n_{l^*,0}$ designates time for executions with reward $0$. UCT guarantees that expected regret grows as the logarithm in the number of rounds (which, for a fixed timeout level, is proportional to execution time). Hence, $n_{l^*,0}\in O(\log(n_{l^*}))$ and $n_{l^*,1}\geq n_{l^*}-O(\log(n_{l^*}))$. Denote by $b$ the number of batches per table. The optimal algorithm executes $b$ batches with timeout $l^*$ and the optimal join order. Skinner can execute at most $m\cdot b-m+1\in O(m\cdot b)$ batches for timeout $l^*$ before no batches are left for at least one table, terminating execution. Since $l^*$ is the smallest timeout greater than the optimal time per batch, the time per batch consumed by Skinner-G exceeds the optimal time per batch at most by factor 2. Hence, denoting by $n^*$ time for an optimal execution, it is $n^*\geq n_{l^*,1}/(2\cdot m)$, therefore $n^*\geq (n_{l^*}-O(\log(n)))/(2\cdot m)\geq n_{l^*}/(2\cdot m)-O(\log(n))$ (since $m$ is fixed), which implies $n^*\geq n/(4\cdot m\cdot\log(n))-O(\log(n))$. Hence, the regret $n-n^*$ is upper-bounded by $(1-1/(4\cdot m\cdot\log(n)))\cdot n+O(\log(n))$. 
\end{proof}

Next, we analyze regret of Skinner-H.


\begin{theorem}
Expected execution time regret of Skinner-H is upper-bounded by $(1-1/(\log(n)\cdot m\cdot 12))\cdot n+O(\log(n))$.
\end{theorem}
\begin{proof}
Denote by $n_O$ and $n_L$ time dedicated to executing the traditional optimizer plan or learned plans respectively. Assuming pessimistically that optimizer plan executions consume all dedicated time without terminating, it is $n_O=\sum_{0\leq l\leq L}2^l$ for a suitable $L\in\mathbb{N}$ at any point. Also, we have $n_L\geq\sum_{0\leq l<L}2^l$ as time is divided between the two approaches. It is $n_L/n\geq (2^L-1)/(2^{L+1}+2^L-2)$ which converges to $1/3$ as $n$ grows. We obtain the postulated bound from Theorem~\ref{skinnerGtheorem} by dividing the ``useful'' (non-regret) part of execution time by factor three.
\end{proof}



The following theorem is relevant if traditional query optimization works well (and learning creates overheads).

\begin{theorem}\label{skinnerHtraditionalTheorem}
The maximal execution time regret of Skinner-H compared to traditional query execution is $n\cdot 4/5$.
\end{theorem}
\begin{proof}
Denote by $n^*$ execution time of the plan produced by the traditional optimizer. Hence, Skinner-H terminates at the latest once the timeout for the traditional approach reaches at most $2\cdot n^*$ (since the timeout doubles after each iteration). The accumulated execution time of all prior invocations of the traditional optimizer is upper-bounded by $2\cdot n^*$ as well. At the same time, the time dedicated to learning is upper-bounded by $2\cdot n^*$. Hence, the total regret (i.e., added time compared to $n^*$) is upper-bounded by $n\cdot 4/5$.
\end{proof}

Finally, we analyze expected regret of Skinner-C.

\begin{theorem}
Expected execution time regret of Skinner-C is upper-bounded by $(1-1/m)\cdot n+O(\log(n))$.\label{skinnerCadditiveTheorem}
\end{theorem}

\begin{proof}
Regret is the difference between optimal execution time, $n^*$, and actual time, $n$. It is $n-n^*=n\cdot(1-n^*/n)$. Denote by $R$ the total reward achieved by Skinner-C during query execution and by $r$ the average reward per time slice. It is $n=R/r$. Denote by $r^*$ the optimal reward per time slice. Reward is calculated as the relative tuple index delta in the left-most table (i.e., tuple index delta in left-most table divided by table cardinality). An optimal execution always uses the same join order and therefore terminates once the accumulated reward reaches one. Hence, we obtain $n^*=1/r^*$. We can rewrite regret as $n-n^*=n\cdot(1-(1/r^*)/(R/r))=n\cdot (1-r/(R\cdot r^*))$. The difference between expected reward and optimal reward is bounded as $r^*-r\in O(\log(n)/n)$~\cite{Kocsis2006}. Substituting $r$ by $r^*-(r^*-r)$, we can upper-bound regret by $n\cdot(1-1/R)+O(\log(n))$. Denote by $R_t\leq R$ rewards accumulated over time slices in which join orders starting with table $t\in T$ were selected. Skinner-C terminates whenever $R_t=1$ for any $t\in T$. Hence, we obtain $R\leq m$ and $n\cdot(1-1/m)+O(\log(n))$ as upper bound on expected regret.
\end{proof}

Instead of the (additive) difference between expected and optimal execution time, we can also consider the ratio.

\begin{theorem}
The ratio of expected to optimal execution time for Skinner-C is upper-bounded and that bound converges to $m$ as $n$ grows. 
\end{theorem}
\begin{proof}
Let $a=n-n^*$ be additive regret, i.e.\ the difference between actual and optimal execution time. It is $n^*=n-a$ and, as $a\leq (1-1/m)\cdot n+O(\log(n))$ due to Theorem~\ref{skinnerCadditiveTheorem}, it is $n^*\geq n-(1-1/m)\cdot n-O(\log(n))=n/m-O(\log n)=n\cdot(1/m-O(\log(n))/n)$. Optimal execution time is therefore lower-bounded by a term that converges to $n/m$ as $n$ grows. Then, the ratio $n/n^*$ is upper-bounded by $m$.
\end{proof}

%% file: sectionsPdf/experiments.tex
\pgfplotsset{
    discard if not/.style 2 args={
        x filter/.code={
            \edef\tempa{\thisrowno{#1}}
            \edef\tempb{#2}
            \ifx\tempa\tempb
            \else
                \def\pgfmathresult{inf}
            \fi
        }
    }
}

\def\addChainPlotTime#1#2{
\addplot table[x index=1, y index=2, col sep=comma, discard if not={0}{CHAINmctsS}] {#1};
\addplot table[x index=1, y index=2, col sep=comma, discard if not={0}{CHAINEddy}] {#1};
\addplot table[x index=1, y index=2, col sep=comma, discard if not={0}{CHAINOpt}] {#1};
\addplot table[x index=1, y index=2, col sep=comma, discard if not={0}{CHAINReopt}] {#1};
\addplot table[x index=1, y index=2, col sep=comma, discard if not={0}{CHAINpostgres}] {#1};
\addplot table[x index=1, y index=2, col sep=comma, discard if not={0}{CHAINpostgresMCTS-Multi-Pure}] {#1};
\addplot table[x index=1, y index=2, col sep=comma, discard if not={0}{CHAINpostgresMCTS-Multi-Hybrid}] {#1};
\addplot table[x index=1, y index=2, col sep=comma, discard if not={0}{CHAINadaptive}] {#1};
\addplot table[x index=1, y index=2, col sep=comma, discard if not={0}{CHAINadaptiveMCTS-Multi-Pure}] {#1};
\addplot table[x index=1, y index=2, col sep=comma, discard if not={0}{CHAINadaptiveMCTS-Multi-Hybrid}] {#1};
\addplot table[x index=1, y index=2, col sep=comma, discard if not={0}{CHAINmonet}] {#1};
\draw [red, ultra thick] (axis cs:\pgfkeysvalueof{/pgfplots/xmin},#2) -- (axis cs:\pgfkeysvalueof{/pgfplots/xmax},#2);
}

\def\addStarPlotTime#1#2{
\addplot table[x index=1, y index=2, col sep=comma, discard if not={0}{STARmctsS}] {#1};
\addplot table[x index=1, y index=2, col sep=comma, discard if not={0}{STAREddy}] {#1};
\addplot table[x index=1, y index=2, col sep=comma, discard if not={0}{STAROpt}] {#1};
\addplot table[x index=1, y index=2, col sep=comma, discard if not={0}{STARReopt}] {#1};
\addplot table[x index=1, y index=2, col sep=comma, discard if not={0}{STARpostgres}] {#1};
\addplot table[x index=1, y index=2, col sep=comma, discard if not={0}{STARpostgresMCTS-Multi-Pure}] {#1};
\addplot table[x index=1, y index=2, col sep=comma, discard if not={0}{STARpostgresMCTS-Multi-Hybrid}] {#1};
\addplot table[x index=1, y index=2, col sep=comma, discard if not={0}{STARadaptive}] {#1};
\addplot table[x index=1, y index=2, col sep=comma, discard if not={0}{STARadaptiveMCTS-Multi-Pure}] {#1};
\addplot table[x index=1, y index=2, col sep=comma, discard if not={0}{STARadaptiveMCTS-Multi-Hybrid}] {#1};
\addplot table[x index=1, y index=2, col sep=comma, discard if not={0}{STARmonet}] {#1};
\draw [red, ultra thick] (axis cs:\pgfkeysvalueof{/pgfplots/xmin},#2) -- (axis cs:\pgfkeysvalueof{/pgfplots/xmax},#2);
}


We evaluate the performance of SkinnerDB experimentally. Additional results can be found in the appendix. 

\subsection{Experimental Setup}


Skinner-G(X) is the generic Skinner version (see Section~\ref{genericSub}) on top of database system X in the following. Skinner-H(X) is the hybrid version on system X. We execute Skinner on top of MonetDB (Database Server Toolkit v1.1 (Mar2018-SP1))~\cite{Boncz2008} and Postgres (version 9.5.14)~\cite{Postgres}. We use different mechanisms to force join orders for those systems. Postgres has dedicated knobs to force join orders. For MonetDB, we ``brute-force'' join orders by executing each join as a separate query, generating multiple intermediate result tables. Skinner-C, described in Section~\ref{customizedSub}, uses a specialized execution engine. We set $w=\sqrt{2}$ in the UCT formula for Skinner-G and Skinner-H and $w=10^{-6}$ for Skinner-C. Unless noted otherwise, we use a timeout of $b=500$ loop iterations for Skinner-C (i.e., thousands or even tens of thousands of join order switches per second). For Skinner-G and -H, we must use much higher timeouts, starting from one second. All SkinnerDB-specific components are implemented in Java. Our current Skinner-C version only allows to parallelize the pre-processing step. Extending our approach to parallel join processing is part of our future work. To separate speedups due to join ordering from speedups due to parallelization, we compare a subset of baselines in single- as well as in  multi-threaded mode. The following experiments are executed on a Dell PowerEdge R640 server with 2 Intel Xeon 2.3~GHz CPUs and 256~GB of RAM. 


\subsection{Performance on Join Order Benchmark}

\begin{table}[t]
\caption{Performance of query evaluation methods on the join order benchmark - single-threaded.\label{jobTable}}
\begin{tabular}{p{1.75cm}p{1.25cm}p{1.25cm}p{1.25cm}p{1.25cm}}
\toprule[1pt]
\textbf{Approach} & \textbf{Total Time} & \textbf{Total Card.\ } & \textbf{Max.\ Time} & \textbf{Max.\ Card.\ }\\
\midrule[1pt]
Skinner-C & 183 & 112M & 9 & 18M \\
\midrule
Postgres & 726 & 681M & 59 & 177M \\
S-G(PG) & 13,348 & N/A & 840 & N/A \\
S-H(PG) & 2,658 & N/A & 234 & N/A \\
\midrule
MonetDB & 986 & 2,971M & 409 & 1,186M \\
S-G(MDB) & 1,852 & N/A & 308 & N/A\\
S-H(MDB) & 762 & N/A & 114 & N/A\\
\bottomrule[1pt]
\end{tabular}
\end{table}

\begin{table}[t]
\caption{Performance of query evaluation methods on the join order benchmark - multi-threaded.\label{jobTableMT}}
\begin{tabular}{p{1.75cm}p{1.25cm}p{1.25cm}p{1.25cm}p{1.25cm}}
\toprule[1pt]
\textbf{Approach} & \textbf{Total Time} & \textbf{Total Card.\ } & \textbf{Max.\ Time} & \textbf{Max.\ Card.\ }\\
\midrule[1pt]
Skinner-C & 135 & 112M & 7 & 18M \\
\midrule
MonetDB & 105 & 2,971M & 26 & 1,186M \\
S-G(MDB) & 1,450 & N/A & 68 & N/A \\
S-H(MDB) &345 & N/A & 86 & N/A \\
\bottomrule[1pt]
\end{tabular}
\end{table}

We evaluate approaches on the join order benchmark~\cite{Gubichev2015}, a benchmark on real, correlated data. We follow the advice of the paper authors and explicitly prevent Postgres from choosing bad plans involving nested loops joins. Tables~\ref{jobTable} and \ref{jobTableMT} compare different baselines in single-threaded and for Skinner, and MonetDB, in multi-threaded mode (our server runs Postgres~9.5 which is not multi-threaded). We compare approaches by total and maximal (per query) execution time (in seconds). Also, we calculate the accumulated intermediate result cardinality of executed query plans. This metric is a measure of optimizer quality that is independent of the execution engine. Note that we cannot reliably measure cardinality for Skinner-G and Skinner-H since we cannot know which results were generated by the underlying execution engine before the timeout.

Clearly, Skinner-C performs best for single-threaded performance. Also, its speedups are correlated with significant reductions in intermediate result cardinality values. As verified in more detail later, this suggests join order quality as the reason. For multi-threaded execution on a server with 24 cores, MonetDB slightly beats SkinnerDB. Note that our system is implemented in Java and does not currently parallelize the join execution phase. 

When it comes to Skinner on top of existing databases, the results are mixed. For Postgres, we are unable to achieve speedups in this scenario (as shown in the appendix, there are cases involving user-defined predicates where speedups are however possible). Postgres exploits memory less aggressively than MonetDB, making it more likely to read data from disk (which makes join order switching expensive). For single-threaded MonetDB, however, the hybrid version  reduces total execution time by nearly 25\% and maximal time per query by factor four, compared to the original system. This is due to just a few queries where the original optimizer selects highly suboptimal plans.

\begin{table}[t]
\caption{Performance of join orders in different execution engines for join order benchmark - single threaded.\label{joinOrdersTable}}
\begin{tabular}{p{1.5cm}p{1.5cm}p{1.75cm}p{1.75cm}}
\toprule[1pt]
\textbf{Engine} & \textbf{Order} & \textbf{Total Time} & \textbf{Max.\ Time} \\
\midrule[1pt]
Skinner & Skinner & 183 & 9 \\
& Optimal & 180 & 7 \\
\midrule
Postgres & Original & 726 & 59 \\
& Skinner & 567 & 14 \\
& Optimal & 555 & 14 \\
\midrule
MonetDB & Original & 986 & 409 \\
& Skinner & 138 & 7 \\
& Optimal & 134 & 6 \\
\bottomrule[1pt]
\end{tabular}
\end{table}

\begin{table}[t]
\caption{Performance of join orders in different execution engines for join order benchmark - multi-threaded.\label{joinOrdersTableMT}}
\begin{tabular}{p{1.5cm}p{1.5cm}p{1.75cm}p{1.75cm}}
\toprule[1pt]
\textbf{Engine} & \textbf{Order} & \textbf{Total Time} & \textbf{Max.\ Time} \\
\midrule[1pt]
Skinner & Skinner & 135 & 7 \\
& Optimal & 129 & 7 \\
\midrule
MonetDB & Original & 105 & 26 \\
& Skinner & 53 & 2.7 \\
& Optimal & 51 & 2.3  \\
\bottomrule[1pt]
\end{tabular}
\end{table}

To verify whether Skinner-C wins because of better join orders, we executed final join orders selected by Skinner-C in the other systems. We also used optimal join orders, calculated according to the $C_{out}$ metric. Tables~\ref{joinOrdersTable} and \ref{joinOrdersTableMT} show that Skinner's join orders improve performance uniformly, compared to the original optimizer. Also, Skinner's execution time is very close to the optimal order, proving the theoretical guarantees from the last section pessimistic. 

\subsection{Further Analysis}

\begin{table}
\caption{Impact of replacing reinforcement learning by randomization.\label{randomizationTable}}
\begin{tabular}{llll}
\toprule[1pt]
\textbf{Engine} & \textbf{Optimizer} & \textbf{Time} & \textbf{Max.\ Time}\\
\midrule[1pt]
Skinner-C & Original & 182 & 9 \\
& Random & 2,268 & 332 \\
\midrule
Skinner-H(PG) & Original & 2,658 & 234 \\
& Random & 3,615 & 250 \\
\midrule
Skinner-H(MDB) & Original & 761 & 114 \\
& Random & $\geq$ 5,743 & $\geq$ 3,600 \\
\bottomrule[1pt]
\end{tabular}
\end{table}

\begin{table}[t]
\caption{Impact of SkinnerDB features.\label{featuresTable}}
\begin{tabular}{p{4.8cm}p{1.25cm}p{1.25cm}}
\toprule[1pt]
\textbf{Enabled Features} & \textbf{Total Time} & \textbf{Max.\ Time} \\
\midrule[1pt]
indexes, parallelization, learning & 135 & 7 \\
parallelization, learning & 162 & 9  \\
learning & 185 & 9 \\
none & 2,268 & 332 \\
\bottomrule[1pt]
\end{tabular}
\end{table}

We experiment with different variants of SkinnerDB. First of all, we compare learning-based selection against randomized selection. Table~\ref{randomizationTable} shows the performance penalty for randomized selection. Clearly, join order learning is crucial for performance. In Table~\ref{featuresTable}, we compare the impact of randomization to the impact of parallelizing pre-processing and adding hash indices on all join columns (which SkinnerDB exploits if the corresponding table is not used in pre-processing). Clearly, join order learning is by far the most performance-relevant feature of SkinnerDB.

\begin{figure}[t]
\subfigure[MonetDB spends most time executing a few expensive queries.]{
\includegraphics{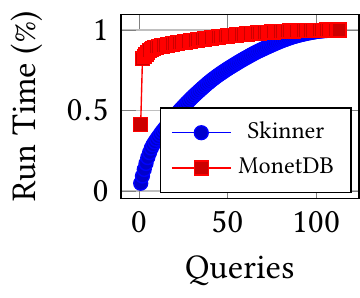}
}
\subfigure[SkinnerDB realizes high speedup for two expensive queries.]{
\includegraphics{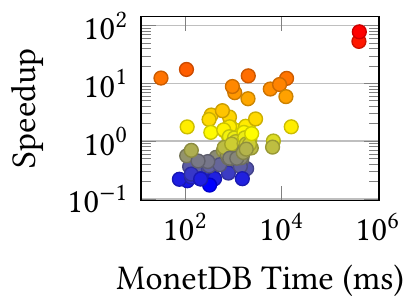}
}
\caption{Analyzing the source of SkinnerDB's speedups compared to MonetDB.\label{monetVsSkinnerFig}}
\end{figure}

We analyze in more detail where the speedups compared to MonetDB come from (all results refer to single-threaded mode). Figure~\ref{monetVsSkinnerFig} shows on the left hand side the percentage of execution time, spent on the top-k most expensive queries (x axis). MonetDB spends the majority of execution time executing two queries with highly sub-optimal join orders (we reached out to the MonetDB team to make sure that no straight-forward optimizations remove the problem). On the right side, we draw speedups realized by Skinner versus MonetDB's query execution time. MonetDB is actually faster for most queries while SkinnerDB has highest speedups for the two most expensive queries. Since those queries account for a large percentage of total execution time, Skinner-C outperforms MonetDB in single-threaded mode. 

\begin{figure}[t]
\subfigure[The growth of the search tree slows down over time.]{
\includegraphics{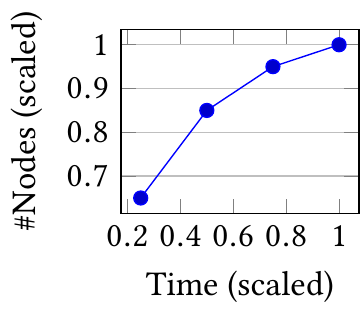}
}
\subfigure[SkinnerDB spends most time executing one or two join orders.]{
\includegraphics{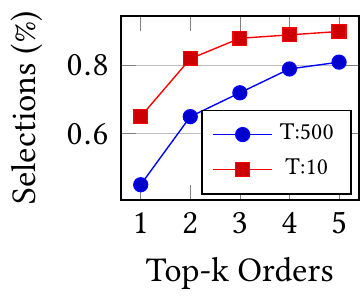}
}
\caption{Analysis of convergence of SkinnerDB.\label{convergenceFig}}
\end{figure}

Figure~\ref{convergenceFig} analyzes convergence of Skinner-C to optimal join orders. On the left side, we show that the growth of the search tree slows as execution progresses (a first indication of convergence). On the right side, we show that Skinner-C executes one (with a timeout of $b=10$ per time slice) or two (with a timeout of $b=500$, allowing less iterations for convergence) join orders for most of the time.

\begin{figure}[t]
\subfigure[Search tree size is correlated with query size.\label{uctMemFig}]{
\includegraphics{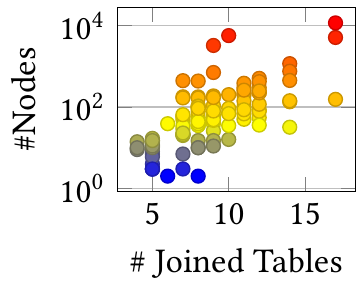}
}
\subfigure[Size of join order progress tracker tree.\label{trackerMemFig}]{
\includegraphics{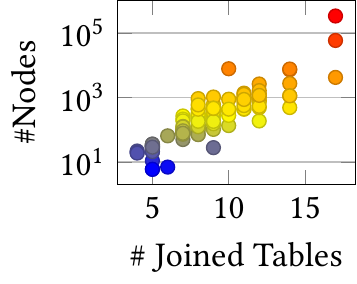}
}
\subfigure[Size of final result tuple indices.\label{finalMemFig}]{
\includegraphics{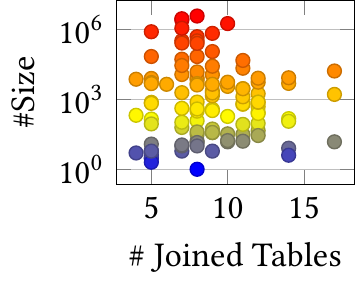}
}
\subfigure[Combined size of intermediate results, progress, and tree.\label{allMemFig}]{
\includegraphics{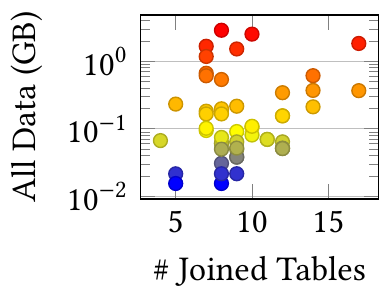}
}
\caption{Memory consumption of SkinnerDB.\label{memoryFigure}}
\end{figure}

Finally, we analyze memory consumption of Skinner-C. Compared to traditional systems, Skinner-C maintains several additional, auxiliary data structures. First, it keeps the UCT search tree. Second, it maintains a tree associating each join order to the last execution state (one tuple index for each base table). Third, it must keep the tuple vectors of all join result tuples in a hash table to eliminate duplicates from different join orders. On the other side, Skinner-C does not maintain any intermediate results as opposed to other systems (due to depth-first multiway join execution). Figure~\ref{memoryFigure} shows the maximal sizes of the aforementioned data structures during query executions as a function of query size. Storing result tuple index vectors (Figure~\ref{finalMemFig}) has dominant space complexity, followed by the progress tracker, and the UCT search tree. Overall, memory consumption is not excessive compared to traditional execution engines.

%% file: sectionsPdf/conclusion.tex
We introduce a new quality criterion for query evaluation strategies: we consider the distance (either difference or ratio) between expected execution time and processing time for an optimal join order. We designed several query evaluation strategies, based on reinforcement learning, that are optimized for that criterion. We implemented them in SkinnerDB, leading to the following insights. First, regret-bounded query evaluation leads to robust performance even for difficult queries, given enough data to process. Second, performance gains by robust join ordering can outweigh learning overheads for benchmarks on real data. Third, actual performance is significantly better than our theoretical worst-case guarantees. Fourth, to realize the full potential of our approach, an (in-query) learning-based optimizer must be paired with a specialized execution engine.

%% file: sectionsPdf/experimentsApp.tex

We show results for additional benchmarks and baselines. As baseline (and underlying execution engine for SkinnerDB), we add a commercial database system ((``ComDB'') with an adaptive optimizer. We also re-implemented several research baselines (we were unsuccessful in obtaining the original code), notably Eddies~\cite{Tzoumas2008} and the Re-optimizer~\cite{Wu2016}. Some of our implementations are currently limited to simple queries and can therefore not be used for all benchmarks. The following experiments are executed on the hardware described before, except for our micro-benchmarks on small data sets which we execute on a standard laptop with 16~GB of main memory and a 2.5~GHZ Intel i5-7200U CPU.

We use an extended version of the \textit{Optimizer Torture Benchmark} proposed by Wu et al. The idea is to create corner cases where the difference between optimal and sub-optimal query plans is significant. \textit{UDF Torture} designates in the following a benchmark with queries where each join predicate is a user-defined function and therefore a black box from the optimizer perspective. We use one good predicate (i.e., join with that predicate produces an empty result) per query while the remaining predicates are bad (i.e., they are always satisfied for the input data). We experiment with different table sizes and join graph structures. \textit{Correlation Torture} is an extended variant of the original benchmark proposed by Wu et al~\cite{Wu2016}. This benchmark introduces maximal data skew by perfectly correlating values in different table columns. As in the original benchmark, we create chain queries with standard equality join and filter predicates. Correlations between predicates and data skew make it however difficult for standard optimizers to infer the best query plan. We vary the position of the good predicate via parameter $m$ between the beginning of the chain ($m=1$) and the middle ($m=nrTables/2$). 

\pgfplotsset{
   /pgfplots/bar  cycle  list/.style={/pgfplots/cycle  list={%
        {fill=blue,mark=none},%
        {fill=blue!50,mark=none},%
        {fill=cyan,mark=none},%
        {fill=blue!70!red,mark=none},%
        {fill=green,mark=none},%
        {fill=green!50,mark=none},%
        {fill=green!70!red,mark=none},%
        {fill=yellow,mark=none},%
        {fill=yellow!50,mark=none},%
        {fill=yellow!70!red,mark=none},
        {fill=black,mark=none}
     }
   },
}

\begin{figure}[t]
\center
\includegraphics{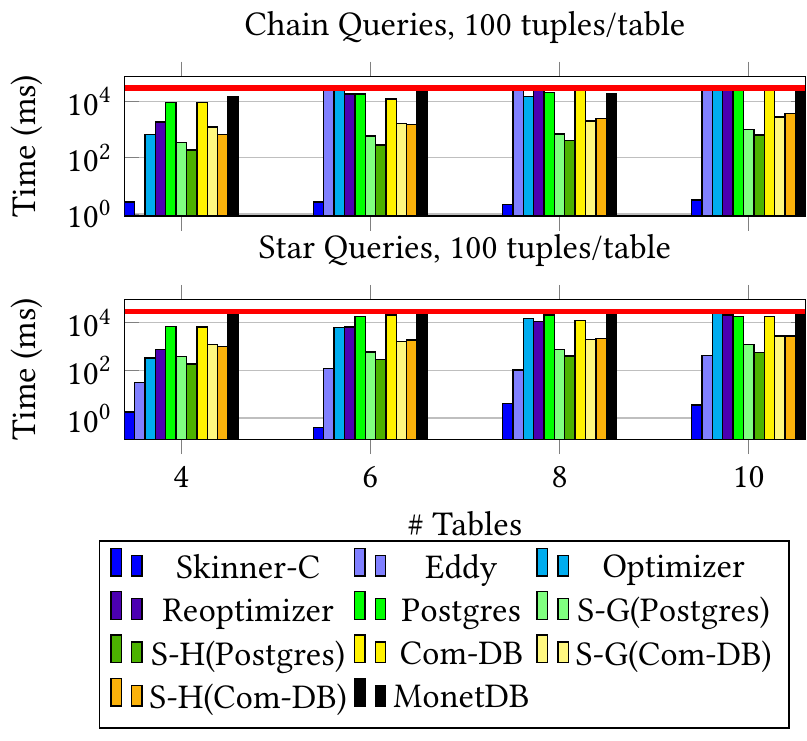}

\caption{UDF Torture benchmark.\label{udfTortureFig}}
\end{figure}

UDF predicates may hide complex code, invocations of external services, or even calls to human crowd workers. They often have to be treated as black boxes from the optimizer perspective which makes optimization hard. Figure~\ref{udfTortureFig} (this and the following figures show arithmetic averages over ten test cases) compares all baselines according to the UDF Torture benchmark described before (the red line marks the timeout per test case). Skinner-C generally performs best in this scenario and beats existing DBMS by many orders of magnitude. We compare a Java-based implementation against highly optimized DBMS execution engines. However, a high-performance execution engine cannot compensate for the impact of badly chosen join orders. Among the other baselines using the same execution engine as we do, Eddy performs best while Optimizer and Re-optimizer incur huge overheads. Re-optimization is more useful in scenarios where a few selectivity estimates need to be corrected. Here, we essentially start without any information on predicate selectivity. For Postgres, our adaptive processing strategies reduce execution time by up to factor 30 for Postgres and large queries. For the commercial DBMS with adaptive optimizer, we achieve a speedup of up to factor 15 (which is in fact a lower bound due to the timeout).

\begin{figure}[t]
\center
\includegraphics{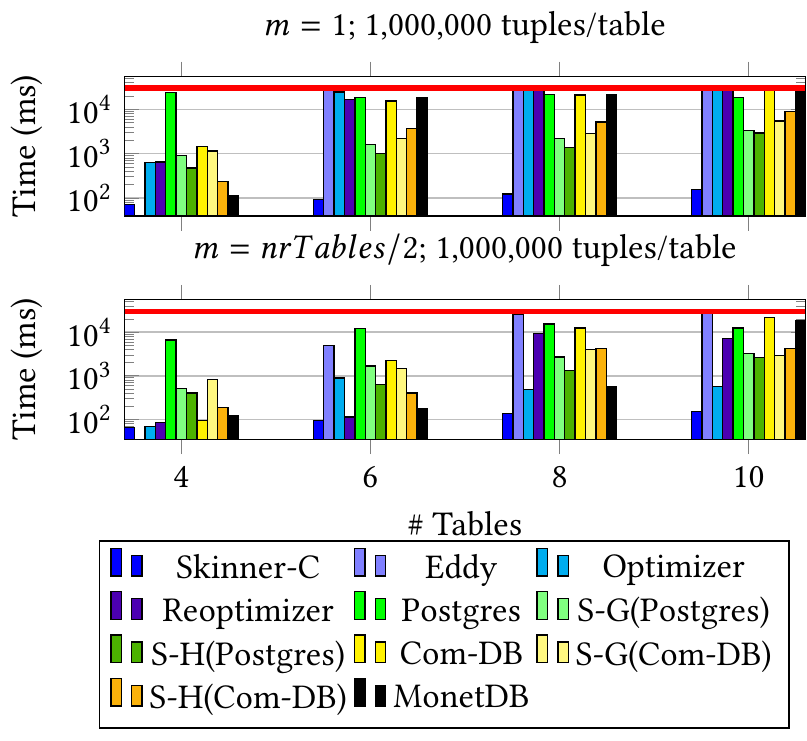}

\caption{Correlation Torture benchmark.\label{correlationLargeFig}}
\end{figure}

Even standard equality predicates can make optimization difficult due to predicate correlations. We evaluate all baselines on the Correlation Torture benchmark~\cite{Wu2016}, Figure~\ref{correlationLargeFig} shows first results. Many of the tendencies are similar to the ones in the UDF Torture benchmark. Skinner-C performs best, traditional query optimizers cope badly with strong predicate correlations. Compared to Figure~\ref{udfTortureFig}, the relative performance gap is slightly smaller. At least in this case, UDF predicates cause more problems than correlations between standard predicates. Again, our adaptive processing strategies improve performance of Postgres and the commercial DBMS significantly and for each configuration (query size and setting for $m$). 

\pgfplotsset{
   /pgfplots/bar  cycle  list/.style={/pgfplots/cycle  list={%
        {fill=blue,mark=none},%
        {fill=blue!50,mark=none},%
        {fill=cyan,mark=none},%
        {fill=blue!70!red,mark=none},%
        {fill=green,mark=none},%
        {fill=green!50,mark=none},%
        {fill=green!70!red,mark=none},%
        {fill=yellow,mark=none},%
        {fill=yellow!50,mark=none},%
        {fill=yellow!70!red,mark=none},
        {fill=black,mark=none}
     }
   },
}

\begin{figure}[h!]
\center
\includegraphics{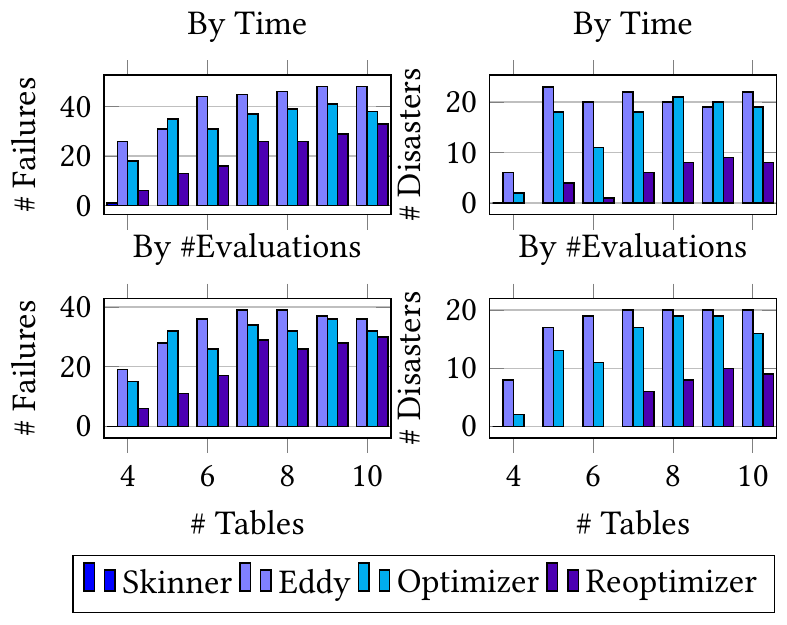}

\caption{Number of ``optimizer failures'' and ``optimizer disasters''.\label{disasterFig}}
\end{figure}

A query evaluation method that achieves bounded overhead in each single case is typically preferred over a method that oscillates between great performance and significant overheads (even if the average performance is the same). Figure~\ref{disasterFig} summarizes results for a new run of the Correlation Torture benchmark, varying number of tables, table size, as well as parameter $m$. We study robustness of optimization and focus therefore on baselines that use the same execution engine. Note that we compare baselines not only with regards to time, but also with regards to the number of predicate evaluations (see lower row) which depends only on the optimizer. We classify for each baseline a test case as \textit{optimizer failure} if evaluation time exceeds the optimum among the other baselines for that test case by factor 10. We call a test case an \textit{optimizer disaster} for factor 100. The figure shows a tight race between Eddy and the traditional optimizer. Re-optimization clearly improves robustness. However, using our regret-bounded algorithms avoids any failures or disasters and is therefore the most robust optimization method. All implementations in Figure~\ref{disasterFig} share code to the extend possible. Still, some of the baselines need to add code that could in principle decrease performance (e.g., per-tuple routing policies for Eddy). To exclude such effects, we also count the number of atomic predicate evaluations for each baseline and re-calculate failures and disasters based on that (bottom row in Figure~\ref{disasterFig}). The tendencies remain the same.

\pgfplotsset{
    discard if not/.style 2 args={
        x filter/.code={
            \edef\tempa{\thisrowno{#1}}
            \edef\tempb{#2}
            \ifx\tempa\tempb
            \else
                \def\pgfmathresult{inf}
            \fi
        }
    }
}

\def\addChainPlotTime#1#2{
\addplot table[x index=1, y index=2, col sep=comma, discard if not={0}{CHAINmctsS}] {#1};
\addplot table[x index=1, y index=2, col sep=comma, discard if not={0}{CHAINEddy}] {#1};
\addplot table[x index=1, y index=2, col sep=comma, discard if not={0}{CHAINOpt}] {#1};
\addplot table[x index=1, y index=2, col sep=comma, discard if not={0}{CHAINReopt}] {#1};
\addplot table[x index=1, y index=2, col sep=comma, discard if not={0}{CHAINpostgres}] {#1};
\addplot table[x index=1, y index=2, col sep=comma, discard if not={0}{CHAINpostgresMCTS-Multi-Pure}] {#1};
\addplot table[x index=1, y index=2, col sep=comma, discard if not={0}{CHAINpostgresMCTS-Multi-Hybrid}] {#1};
\addplot table[x index=1, y index=2, col sep=comma, discard if not={0}{CHAINadaptive}] {#1};
\addplot table[x index=1, y index=2, col sep=comma, discard if not={0}{CHAINadaptiveMCTS-Multi-Pure}] {#1};
\addplot table[x index=1, y index=2, col sep=comma, discard if not={0}{CHAINadaptiveMCTS-Multi-Hybrid}] {#1};
\addplot table[x index=1, y index=2, col sep=comma, discard if not={0}{CHAINmonet}] {#1};
\draw [red, ultra thick] (axis cs:\pgfkeysvalueof{/pgfplots/xmin},#2) -- (axis cs:\pgfkeysvalueof{/pgfplots/xmax},#2);
}

\def\addStarPlotTime#1#2{
\addplot table[x index=1, y index=2, col sep=comma, discard if not={0}{STARmctsS}] {#1};
\addplot table[x index=1, y index=2, col sep=comma, discard if not={0}{STAREddy}] {#1};
\addplot table[x index=1, y index=2, col sep=comma, discard if not={0}{STAROpt}] {#1};
\addplot table[x index=1, y index=2, col sep=comma, discard if not={0}{STARReopt}] {#1};
\addplot table[x index=1, y index=2, col sep=comma, discard if not={0}{STARpostgres}] {#1};
\addplot table[x index=1, y index=2, col sep=comma, discard if not={0}{STARpostgresMCTS-Multi-Pure}] {#1};
\addplot table[x index=1, y index=2, col sep=comma, discard if not={0}{STARpostgresMCTS-Multi-Hybrid}] {#1};
\addplot table[x index=1, y index=2, col sep=comma, discard if not={0}{STARadaptive}] {#1};
\addplot table[x index=1, y index=2, col sep=comma, discard if not={0}{STARadaptiveMCTS-Multi-Pure}] {#1};
\addplot table[x index=1, y index=2, col sep=comma, discard if not={0}{STARadaptiveMCTS-Multi-Hybrid}] {#1};
\addplot table[x index=1, y index=2, col sep=comma, discard if not={0}{STARmonet}] {#1};
\draw [red, ultra thick] (axis cs:\pgfkeysvalueof{/pgfplots/xmin},#2) -- (axis cs:\pgfkeysvalueof{/pgfplots/xmax},#2);
}

\begin{figure}
\center
\includegraphics{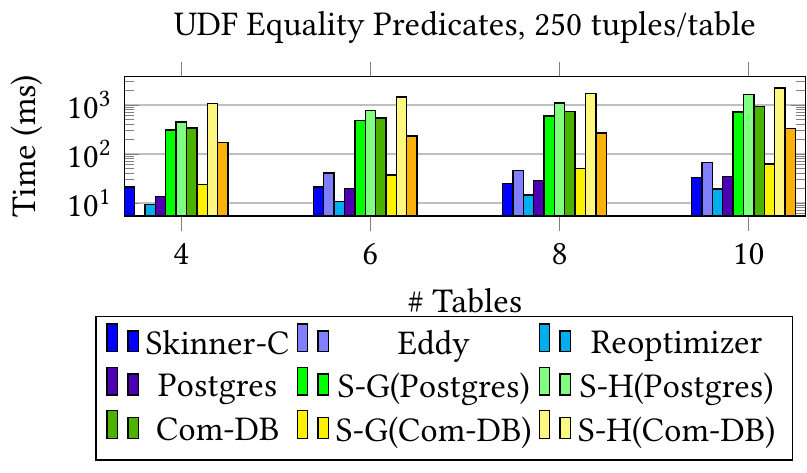}

\caption{Trivial Optimization benchmark.\label{easyFig1}}
\end{figure}

Our primary goal is to achieve robust query evaluation for corner cases. Still, we also consider scenarios where sophisticated optimization only adds overheads. Figure~\ref{easyFig1} shows results for the Trivial Optimization benchmark in which all query plans avoiding Cartesian products are equivalent. We are mostly interested in relative execution times obtained for the same execution engine with different optimization strategies. Clearly, optimizers that avoid any exploration perform best in this scenario. For the four baselines sharing the Java-based execution engine (Optimizer, Re-Optimizer, and Eddy), this is the standard optimizer. For the baselines that are based on existing DBMS, the original optimizer works best in each case. While robustness in corner cases clearly costs peak performance in trivial cases, the overheads are bounded.


\pgfplotsset{
   /pgfplots/bar  cycle  list/.style={/pgfplots/cycle  list={%
        {fill=blue,mark=none},%
        {fill=green,mark=none},%
        {fill=green!50,mark=none},%
        {fill=green!70!red,mark=none},%
        {fill=black,mark=none}
     }
   },
}

\begin{figure}[t]
\center
\includegraphics{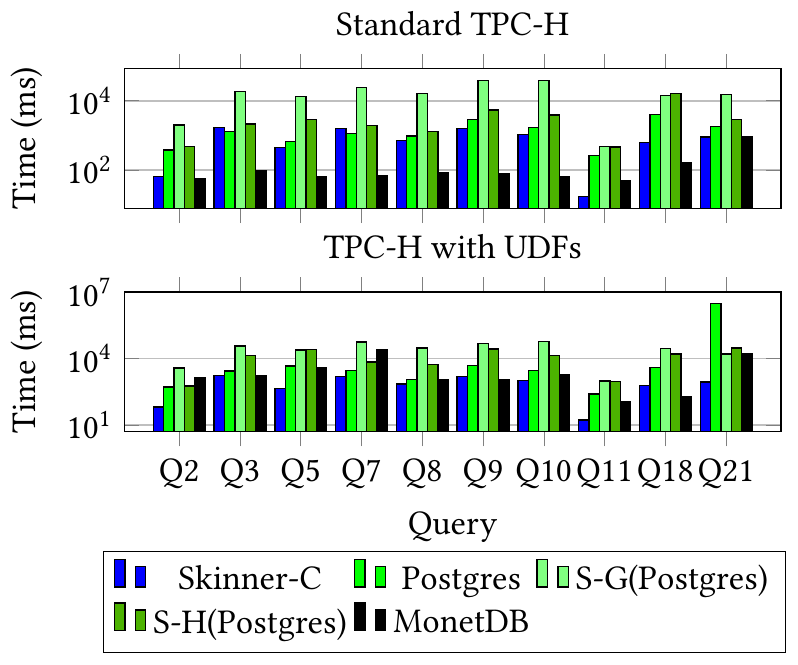}

\caption{Performance on TPC-H queries.\label{tpchFig}}
\end{figure}

Finally, we benchmark several baselines on the more complex queries of the TPC-H benchmark~\cite{TPC2013}. We restrict evaluated approaches to the ones where our current implementation supports the full set of TPC-H queries. Figure~\ref{tpchFig} reports processing times of ten TPC-H queries that join at least three tables. For each query and each approach, we calculate the relative overhead (i.e., query execution time of approach divided by execution time of best approach for this query). The ``Max.\ Rel.\'' column contains for each approach the maximal value over all queries. We study original TPC-H queries as well as a variant that makes optimization hard. The latter variant replaces all unary query predicates by user-defined functions. Those user-defined functions are semantically equivalent to the original predicate. They typically increase per-tuple evaluation overheads. Most importantly, however, they prevent the optimizer from choosing good query plans.

\begin{table}[b]
\caption{Result summary for TPC-H variants.\label{tpchTable}}
\begin{tabular}{llll}
\toprule[1pt]
\textbf{Scenario} & \textbf{Approach} & \textbf{Time (s)} & \textbf{Max.\ Rel.\ }\\
\midrule[1pt]
TPC-H & Skinner-C & 9 & 22 \\
\midrule
& Postgres & 15 & 37 \\
\midrule
& S-G(Postgres) & 182 & 594 \\
\midrule
& S-H(Postgres) & 38 & 97 \\
\midrule
& MonetDB & 2 & 3 \\
\midrule[1pt]
TPC-UDF & Skinner-C & 9& 3 \\
\midrule
& Postgres & 3,117 & 3,457 \\
\midrule
& S-G(Postgres) & 305 & 154 \\
\midrule
& S-H(Postgres) & 142 & 88 \\
\midrule
& MonetDB & 53 & 20 \\
\bottomrule[1pt]
\end{tabular}
\end{table}

The upper half of Figure~\ref{tpchFig} shows results on original TPC-H queries while the lower half reports on the UDF variant. Table~\ref{tpchTable} summarizes results, reporting total benchmark time as well as the maximal per-query time overhead (compared to the optimal execution time for that query over all baselines). MonetDB is the clear winner for standard queries (also note that MonetDB and SkinnerDB are column stores while Postgres is a row store). SkinnerDB achieves best performance on the UDF variant. Among the three Postgres-based approaches, the original DBMS performs best on standard cases. The hybrid approach performs reasonably on standard cases but reduces total execution time by an order of magnitude for the UDF scenario. We therefore succeed in trading peak performance in standard cases for robust performance in extreme cases.
\balance
